\documentclass[12pt,a4paper,reqno,oneside]{amsart}
\usepackage{t1enc}
\usepackage{times}
\usepackage{amssymb}
\usepackage[mathscr]{euscript}
\usepackage{graphicx}
\usepackage{hyperref}
\usepackage{enumitem}
\usepackage[final]{changes}
\usepackage{bbm}
\usepackage[foot]{amsaddr}
\colorlet{Changes@Color}{red}

\newtheorem{thm}{Theorem}[section]
\newtheorem{prop}[thm]{Proposition}
\newtheorem{lem}[thm]{Lemma}
\newtheorem{cor}[thm]{Corollary}

\theoremstyle{remark}
\newtheorem{rem}[thm]{Remark}

\newtheorem{ex}[thm]{Example}

\theoremstyle{definition}
\newtheorem{defn}[thm]{Definition}

\renewcommand{\phi}{\varphi} 

\newcommand{\E}{\mathbb{E}} 

\makeatletter
\DeclareRobustCommand\bfseries{%
	\not@math@alphabet\bfseries\mathbf
	\fontseries\bfdefault\selectfont
	\boldmath 
}
\makeatother

\makeatletter
\renewcommand{\email}[2][]{%
	\ifx\emails\@empty\relax\else{\g@addto@macro\emails{,\space}}\fi%
	\@ifnotempty{#1}{\g@addto@macro\emails{\textrm{(#1)}\space}}%
	\g@addto@macro\emails{#2}%
}
\makeatother
\begin{document}
 
\title{Statistically consistent term structures have affine geometry}

\begin{abstract}
	
This paper is concerned with finite dimensional models for the entire term structure for energy futures. As soon as a finite dimensional set of possible yield curves is chosen, one likes to estimate the dynamic behaviour of the yield curve evolution from data. The estimated model should be free of arbitrage which is known to result in some drift condition. If the yield curve evolution is modelled by a diffusion, then this leaves the diffusion coefficient open for estimation. From a practical perspective, this requires that the chosen set of possible yield curves is compatible with any obtained diffusion coefficient. In this paper, we show that this compatibility enforces an affine geometry of the set of possible yield curves.

\smallskip
\noindent \textbf{Keywords:} Electricity futures models, Finite dimensional realisation, Affine geometry, Statistically consistency.

\smallskip
\noindent \textbf{MSC(2010):} 91B24, 91G20.
\end{abstract}

%
\author[Kr\"uhner]{Paul Kr\"uhner$^\ast$}
\address[$\ast$]{Institute for Statistics and Mathematics, WU Wien, Vienna, Austria. Email: paul.eisenberg@wu.ac.at}

\author[Xu]{Shijie Xu$^\dagger$}
\address[$\dagger$]{\textit{(Corresponding author)} Institute for Financial and Actuarial Mathematics, University of Liverpool, Liverpool, Uk. Email: ShijieXu@liverpool.ac.uk}



\maketitle

\section{Introduction}

Since energy markets in some countries around the world are liberated, trading in energy has gained importance for non-energy traders. Energy exchanges, e.g.\ European Energy Exchange (EEX), Scandinavian Nordic Power Exchange (Nord Pool) and Chicago Mercantile Exchange (CME), are not only interesting to traditional participants like electricity producers, retailers and consumers but also to institutional investors like pension funds and hedge funds, cf.\ Benth, Benth and Koekebakker \cite{benth2008stochastic}. 




The European Energy Exchange (EEX) reports electricity prices for spot markets, intraday markets and futures markets. Our paper focuses on the modelling of energy futures markets. The three main approaches to model energy futures prices in the literature are: to derive their prices from supply-demand equilibrium modelling or spot modelling or to model them directly via the Heath-Jarrow-Morton approach.

Among these, the most common approach is supply-demand equilibrium modelling which is followed by the balancing restriction. It aims at modelling existing (sometimes future) production capacity and the amount of consumption required at any time. Typically, this is used to explain hourly electricity spot prices via supply and demand matching. B\"uhler and M\"uller-Merbach  \cite{buhler2007dynamic}  propose a dynamic competitive equilibrium model for pricing electricity futures.  

The second most common approach is the so-called spot modelling. It is ignorant to supply and demand but models the spot price directly as a stochastic process. Lucia and Schwartz \cite{Lucia} have studied one and two-factor models for the valuation of power derivatives. It can be used for spot prices forecasts based on past observations or to explain derivatives written on the energy market, e.g.\ futures. However, there is no simple link between futures prices and spot prices.  Additional to spot prices forecasts, one requires to introduce a risk premium $R$, i.e.\
$$F_t(T_1,T_2) = \frac{1}{T_2-T_1} \int_{T_1}^{T_2} \E[ r_u | \mathcal F_t] du + R(t,T_1,T_2), \quad 0\leq t\leq T_1 \leq T_2,$$
(or, alternatively, the risk premium $R$ equals $0$ under the pricing measure $\mathbb Q$, see Benth, Kallsen and Meyer-Brandis  \cite{benth2007non} for a more systematic approach and overview to spot modelling.)

Another approach is the so-called Heath-Jarrow-Morton (HJM) approach which we adopt in this article. This approach has the advantage of incorporating futures price data more directly into the model but does not explain futures prices in a deeper way. One models an artificial instantaneous quantity, the risk-neutral price forecast of one unit of energy delivered at a fixed future point of time. Roughly, the left-hand side of
$$f_t(x)  = \E_{\mathbb Q}[ r_{x+t} | \mathcal F_t], $$
is the directly modelled price where $r$ is the price process of a hypothetically existing underlying spot model and $\mathbb Q$ is an unknown pricing measure. Actual futures prices are recovered via
$$F_t(T_1,T_2) = \frac{1}{T_2-T_1} \int_{T_1}^{T_2} f_t(u-t)  du, \quad 0\leq t\leq T_1 \leq T_2.$$

This approach is inspired by the analogous approach for money markets, cf.\ Heath, Jarrow and Morton \cite{HJM} and is easily adopted for energy markets. It has been followed by a number of authors. Koekebakker and Ollmar \cite{koekebakker2005forward} specify a multi-factor term structure model in an {HJM}-framework to examine the forward curve dynamics in the Nordic electricity market. Benth and Kr\"uhner  \cite{benth2018approximation} find that basically any {HJM}-type model can be approximated by a finite dimensional model with no-arbitrage. Further, they show that the forwards can be represented as the sum of the spot price and the transformation of several generalised Ornstein-Uhlenbeck processes, which gives a new perspective on the findings of Koekebakker and Ollmar \cite{koekebakker2005forward}.

In this paper, we adopt the HJM approach for the instantaneous quantity $f_t$ and additionally ask for several desirable properties of the model. 
We will make these properties precise in Section \ref{s:MainResults}.
\begin{enumerate}
  \item We ask the model to be driven by a finite dimensional process via
      $$ f_t(x) = g(x,Y_t)$$ 
 where $Y$ is some finite dimensional process and $g:\mathbb R_+\times \mathbb R^d\rightarrow \mathbb R$ some deterministic given function. We call such a representation for $f$ finite dimensional realisation (FDR). This approach has been analysed in B\"uhler \cite{buehler2006consistent} for the modelling of variance swap prices.
 \item Absence of arbitrage (NA).
 \item Additionally, to (1), (2) the model should be compatible with any possible diffusion coefficient for $Y$.
\end{enumerate}

The first property is useful for computer implementation and has essentially be analysed in B\"uhler \cite{buehler2006consistent}. Absence of arbitrage is a theoretically desirable property which usually leads to some kind of drift condition, see \cite[Theorem 2.2]{buehler2006consistent}. If one has chosen a function $g$, then one would estimate the diffusion coefficient (for instance as a function of the state of $Y$) from data. Then one would use the drift from the drift condition (implied by (2)) and the process $Y$ is, hopefully, completely determined. As it turns out, whether this is possible for any resulting diffusion coefficient or not depends on the structure of the function $g$. We will show that this consistency condition implies that $g$ has the following affine structure:
\begin{align}\label{affine-choice}
g(x,y) = c(x) + u(x)A(y), \quad x\geq 0,\ y\in \mathbb R^d,
\end{align}
for some functions $c$, $u$ and $A$. This in turn means that, if one wants to choose possible forward curves $g(\cdot,y), y\in \mathbb R^d$ and be free to estimate the diffusion coefficient of $Y$ from data while keeping absence of arbitrage, then $g$ has to be affine as above. 


We will refer to (3) as statistically consistency condition ({SCC}) and this is made precise in Definition \ref{SCC} below.

In Example \ref{example2} below we show that term structure models for energy futures do not need to have affine geometry. That is, FDR models do not need to be affine. In a related work, Tappe \cite{Tappe} shows that {FDR} models of L\'evy-driven stochastic partial differential equations ({SPDE}) with arbitrarily small jumps have affine geometry.  Bj\"ork \cite{bjork} has shown that {FDR} models for interest rate markets which have affine geometry are of quasi-exponential form. This result readily carries over to futures modelling for energy markets.

This paper is organised as follows. In Section \ref{s:MainResults} we introduce our setting, recall some essentially known results and formulate our main result Theorem \ref{Main}. In Section \ref{s:example} we provide an example for the necessity of SCC. In the appendix, we give some technical proofs. 

\subsection{Notations} 

		Let $U\subseteq \mathbb R^d$ open and $g:U\rightarrow \mathbb R$ be differentiable, we write	
		$$ \nabla_yg(x,y):= \left(\partial_{y_1}g(x,y),\dots, \partial_{y_d}g(x,y)\right), $$
		for the \emph{gradient} of $g(x,\cdot)$ at position $y\in U$.
		
		 For $x,y\in\mathbb R^d$ we define the standard \emph{scalar product} $xy:=\sum_{j=1}^d x_j y_j$.
		 
	   	 We denote the \emph{Euclidean norm} for $x\in \mathbb R^d$ by $|x| := \sqrt{\sum_{j=1}^d x_j^2}$. 
	   	 
		 Throughout the entire paper $\left(\Omega, (\mathcal F_t)_{t\geq 0}, \mathbb P\right)$ denotes a filtered \emph{probability space} where $\mathcal F_0$ contains all $\mathbb P$-null sets and $\mathcal F$ is right-continuous. Also, we denote the \emph{Lebesgue measure} on $\mathbb{R}^d$ by $\lambda_d$. 
		 
		For an $\mathbb R^d$-valued stochastic process $X$ we say that $X$ has a \emph{drift coefficient} $b$ where $b$ is progressively measurable and pathwise integrable if
		$$ X_t - \int_0^tb_s ds,\quad t\geq 0, $$
		defines a local martingale and in this case we write $\mathrm{drift}(X) = b_t dt $.

\section{Main results}\label{s:MainResults}

Electricity forward traded in the real world are contracts prescribe power and delivery time interval in the future. We denote by $F_t(T_1,T_2)$, the discounted price at time $t \geq 0$ of a forward contract delivering $1$ unit of energy with equal delivery rate over the time interval $[T_1,T_2]$. 

Based on the real contract traded in the market, we introduce the instantaneous forward price. The instantaneous forward price $F_t(T)$ is the price at time $t$ of a forward contract delivering $1$ unit of energy instantly at time $T$. This artificial price typically  only exists in the model. 

The price at time $t$ of a forward contract delivering $1$ unit of energy from $T_1$ to $T_2$ is the average forward during the delivery period,
	$$F_t(T_1,T_2)=\frac{1}{T_2-T_1}\int_{T_1}^{T_2}F_t(\mu)d\mu,\quad  0\leq t\leq T_1\leq T_2.$$
	
We look at futures prices at a fixed point of time $t$, e.g.\ now. 
Artificially, we pretend that we have the price for all futures with all possible delivery periods in the market (which we call full term structure). Then we can determine the price for one unit of energy delivered instantaneously via
$$  F_t(T)  = \lim_{h\searrow0 } F_t(T,T+h). $$ 

Throughout this paper, we will transfer {HJM}-equations into the Musiela parametrisation (or we say Heath-Jarrow-Morton-Musiela ({HJMM}) equations, see Brace and Musiela \cite{brace1994multifactor}), i.e.\ we observe the price for instantaneous delivery futures in time to maturity rather than time of maturity, and we define
$$f_t(x):=F_t(t+x), \quad t\geq 0, $$
where $x\geq 0$ is the time to delivery. By a model for the futures price we mean a function
$ f:\Omega\times \mathbb R_+\times \mathbb R_+ \rightarrow \mathbb R $
where
$ (f_{\omega,t}(x))_{\omega\in\Omega, t\geq 0} $
is a progressively measurable stochastic process for any $x\geq 0$ and $x\mapsto f_t(x)$ is locally integrable for any fixed $t\geq 0$.

We like to prescribe the possible curves $x\mapsto f_t(x)$, which we can see in the market in an as simple as possible mathematical way. This is important for various reasons:
\begin{enumerate}
	\item Simple descriptions are easier to communicate;
	\item Calibration is easier if the mathematical structure is as simple as possible;
	\item Simulations are less involved.
\end{enumerate}

The simplest but somewhat generic method of prescribing possible curves is to say that we have a family of possible curves indexed by some element $y$, i.e.\ 
$$  f_t(x) = g(x,y), \quad x\geq0 ,$$
for some $y$ depending on $t\geq 0$, $\omega\in\Omega$. We will further assume that $y$ lies in $\mathbb R^d$, i.e.\ we have a family of possible market curves indexed by the $d$-dimensional vector space.

We will assume that the dependence on the parameter $y$ is “smooth” and the dependence on the parameter $x$ is “smooth” as well. In fact we assume that $g$ is at least of class $C^{1,2}$. (e.g.\ $g(x,y) = y e^{-x}$ for $y\in\mathbb R$, or $g(x,y) = y_1 e^{-xy_2}$ for $y\in\mathbb R^2$.)

For each point of time $t$ and each scenario $\omega\in \Omega$ we see a curve  $x\mapsto f_t(x)$ and find the corresponding (random) parameter $Y_t$ and obtain
$$ f_t(x) = g(x,Y_t), \quad t\geq 0, \ x\geq 0. $$

Finally, we will assume that $Y$ is an It\^o process as well. 

The notion of FDR is originating from the following problem. Given a process $f$ with value in a Hilbert space, when does it stay on some finite dimensional sub-manifold $M$ with $f_t=g(Y_t)$ locally for some finite dimensional process $Y$? This question has been analysed in \cite[Theorem 13]{teichmann2005stochastic}. We start with such a representation $f_t=g(Y_t)$. This is the point of view taken in \cite[Definition 2.4]{buehler2006consistent} in the case of variance curves. 

This leads to the following definition:

\begin{defn}[Finite dimensional realisation]\label{FDR}
	We say that a model $f$ has \textit{finite dimensional realisation} ({FDR}) if
	$$f_t(x)=g(x,Y_t),$$
	where $g : \mathbb R_+\times \mathbb R^d \rightarrow \mathbb R$ is a deterministic $C^{1,2}$-function and $Y$ is a finite dimensional stochastic process such that
	$$ dY_t = \beta_t dt + \sigma_t dW_t, \quad t\geq 0, \sigma>0,$$
	for some Brownian motion $W$ where $\beta$ (resp.\ $\sigma $) are $d$-dimensional (resp.\ $d\times d$-dimensional) progressively measurable process such that  
	$$ \int_0^t |\beta_s| ds < \infty ,\quad \int_0^t |\sigma_s|^2 ds <\infty,$$
	and such that for any bounded open set $U\subseteq\mathbb R$ we have $\beta_{\cdot} 1_{\{|Y_{\cdot}|\in U \}}$ and $\sigma_{\cdot} 1_{\{|Y_{\cdot}|\in U \}}$ are bounded in $t$ and $\omega$.
\end{defn}

Motivated by Teichmann, Klein and Cuchiero \cite{Cuchiero}, we use the following notion of risk neutrality:
\begin{defn}[Risk-neutral]\label{t:NA}
	We say that a model $f$ is \textit{risk-neutral} ({RN}) if
	\begin{enumerate}	
		\item \label{NA:def:1} $x \mapsto f_t(x)$ is a $C^1$-function for any $t\geq 0$;
		\item \label{NA:def:2} $\mathrm{drift}(f(x)) = \partial_x f_t(x) dt$ for any $x\geq 0$.
	\end{enumerate}

	We say that a model $f$ does not allow for arbitrage ({NA}) if there is an equivalent measure $\mathbb Q$ such that it is risk neutral under $\mathbb Q$.
\end{defn}

\begin{defn}[Affine geometry]\label{d:flat}
  We say that an FDR model $f_t(x) = g(x,Y_t)$ is of \textit{affine geometry} (or \textit{affine-like}) if there are functions $c:\mathbb R_+\rightarrow\mathbb R$, $u:\mathbb R_+\rightarrow \mathbb R^d$ and $A:\mathbb R^d\rightarrow \mathbb R^d$ such that
   \begin{align}\label{eq:affine geometry}
   	g(x,y) = c(x) + u(x)A(y),\quad x\geq 0,y\in\mathbb R^d. 
   \end{align} 
\end{defn}

\begin{rem}
In the above definition the name "affine" refers to the geometric shape of the range of $G:\mathbb R^d\rightarrow \mathbb R^{\mathbb R_+}, y \mapsto g(\cdot,y)$. If an FDR model is affine-like, then the possible curves are all in the affine space $c + \mathrm{Span}(u_1,\dots,u_d)$.

\end{rem}

\begin{rem}
From an economical perspective, $\left( F_t(T_1,T_2) \right) _{0\leq t\leq T_1 \leq T_2}$ are prices for the contract which are actually on the market. The fundamental theorem of asset pricing, which is stated by Delbaen and Schachermayer \cite{Delbaen}, showed that these are local martingales under some equivalent measures if and only if there is no-arbitrage in the model. Teichmann, Klein and Cuchiero \cite{Cuchiero} extended the criterion to large markets, i.e.\ markets with possibly infinitely many securities.

If $f$ does not allow for arbitrage in our sense, then $(F_{\cdot}(T_1,T_2))_{0< T_1\leq T_2}$  are local martingales on $[0, T_1]$. (See Proposition \ref{NA} below.)
Conversely, if $(F_{\cdot}(T_1,T_2))_{0< T_1\leq T_2}$ are local martingales, one would expect that under reasonable conditions, for any fixed $T\geq 0$,  $F_{\cdot}(T)$ is a local martingale under some equivalent measures which yields that \eqref{NA:def:2} in the above definition must be met under some equivalent measure. 

We proceed with the risk-neutral view, i.e.\ we pretend that $\mathbb P$ is already the local martingale measure.
\end{rem}


			
If we assume a finite dimensional realisation, then we can characterise risk neutrality in various ways.  The following Proposition is in the spirit of Buehler's result \cite[Theorem 2.2]{buehler2006consistent} for variance curve models with slightly different assumptions. For the sake of completeness we give a proof.

\begin{prop}[Risk neutrality condition]\label{NA}
Assume that FDR holds. Then the following statements are equivalent:
\begin{enumerate}
	\item \label{NA:prop:1} RN holds;
	\item \label{NA:prop:2} There is a $\lambda_1\otimes \mathbb P$-null set $N\subseteq \mathbb R_+\times\Omega$ such that
	\begin{align}\label{eq:general}
	\partial_xg(x,Y_t) = \nabla_y g(x,Y_t)\beta_t + \frac12 \sum_{i,j=1}^d \sigma_{t,i,j}\sigma_{t,j,i} \partial_{y_i}\partial_{y_j}g(x,Y_t),
	\end{align}
	for any $x\geq 0$, $(t,\omega)\in (\mathbb R_+\times\Omega)\setminus N$;
	\item $(F_t(T_1,T_2))_{t\in[0,T_1]}$ is a local martingale for any $0<T_1<T_2$.
\end{enumerate}
\end{prop}
\begin{proof}
We define $a_t(x) := \nabla_yg(x,Y_t)\beta_t + \frac12 \sum_{i,j=1}^d \sigma_{t,i,j}\sigma_{t,j,i}\partial_{y_i}\partial_{y_j}g(x,Y_t)$ for $t,x\geq 0$. It\^o's formula yields
$$ \mathrm{drift}(f(x)) = a_t(x) dt,$$
for any $x\geq 0$. Let $0<T_1<T_2$ and define the $C^{1,2}$-function 
$$G:[0,T_1]\times\mathbb R^d\rightarrow\mathbb R, \quad (t,y) \mapsto \frac{1}{T_2-T_1}\int_{T_1}^{T_2} g(u-t,y) du.$$

We find that $ F_t(T_1,T_2) = G(t,Y_t)$ for any $t\geq 0$ and It\^o's formula yields
$$ \mathrm{drift}(F(T_1,T_2)) = \frac{1}{T_2-T_1}\int_{T_1}^{T_2} \left(a_t(u-t) - \partial_x g(u-t,Y_t) \right) du. $$

$\underline{(1)\Rightarrow (2)}$: Assume RN holds. Then $\partial_x g(x,Y_t) = \partial_x f_t(x) = a_t(x)$,  $\lambda_1\otimes \mathbb P$-a.s.\ as claimed.

$\underline{(2)\Rightarrow (3)}$: Assume $a_t(x) = \partial_x f_t(x)$. Then we find that
$$ \mathrm{drift}(F(T_1,T_2)) = 0 dt,$$
and hence $(F_t(T_1,T_2))_{t\in [0,T_1]}$ is a local martingale.

$\underline{(3)\Rightarrow (1)}$: Assume that $(F_t(T_1,T_2))_{t\in [0,T_1]}$ is a local martingale for any $0<T_1<T_2$. Then we have
$$ \frac{1}{T_2-T_1}\int_{T_1}^{T_2} \left(a_t(u-t) - \partial_x f(u-t,Y_t) \right) du = 0,$$
for any $0\leq t\leq T_1<T_2$ where the exception null set depends on $T_1,T_2$. However, continuity of $T_1,T_2$ yields a single $\lambda_1\otimes \mathbb P$-null set $N\subseteq \mathbb R_+\times \Omega$ outside which we have
$$ \int_{T_1}^{T_2} \left(a_t(u-t) - \partial_x f(u-t,Y_t) \right) du = 0.  $$

Consequently, we find outside a possibly larger null set that
$$ a_t(x) = \partial_x f(x,Y_t),$$
as required and $x\mapsto f_t(x)$ is $C^1$ anyway because we assumed FDR.
\end{proof}

Now, we introduce the statistically consistency condition the main concept in our paper. Basically, the statistically consistency condition for a given function $g$ means that for every possible constant volatility, there is an It\^o process $Y$ with the given volatility which will lead to an arbitrage-free model if $f_t(x)=g(x, Y_t)$.
This is motivated from the fact that we would like to estimate the volatility from data and our model should work with whatever volatility we get out of the estimator.
\begin{defn}[Statistically consistency condition]\label{SCC}
	We assume that we have a model with {FDR}
	$$f_t(x) = g(x,Y_t),\quad x,t\geq 0.$$
	
	We say that the \textit{statistically consistency condition} (or {SCC}) holds, if for any $\sigma\in\mathbb R^{d\times d}$ there is an $\mathbb R^d$-valued progressively measurable process $\beta^\sigma$ with
	\begin{enumerate}
		\item $\int_0^t|\beta_s|ds<\infty, a.s.$ for any $t\geq 0$.
		\item \label{SCC:RN} The model $$f^{\sigma,y_0}_t(x)  := g(x,Y_t^{\sigma,y_0}), \quad  t,x\geq 0,$$
		defines an RN model where $Y_t^{\sigma,y_0} := y_0 + \int_0^t \beta^{\sigma,y_0}_s ds + \sigma W_t$, $t\geq 0$.
		\item \label{SCC:def:3} The random variables $\partial_{y_i}g(x,Y_t^{\sigma,y_0})$, $\partial_{y_i}\partial_{y_j}g(x,Y_t^{\sigma,y_0})$, $\partial_{y_i}g(x,Y_t^{\sigma,y_0})\beta^{\sigma,y_0}_t$ and $\beta^{\sigma,y_0}_t$ have finite absolute expectation for any $i,j=1,\dots,d$, $x,t\geq 0$.
		\item \label{locB}For any bounded set $B\subseteq\mathbb R^d$ there is $C>0$ such that $|\beta^{\sigma,y_0}_t|1_{\{Y^{\sigma,y_0}\in B\}} \leq C$.
	\end{enumerate}
	
	In the following we will often suppress the starting value $y_0$ and simply write $\beta^\sigma$, $Y^\sigma$ and $f^\sigma$ respectively.
\end{defn}

\begin{rem}
	In \eqref{SCC:RN} in the preceding definition it is reasonable to ask for the seemingly weaker condition NA instead of RN. However, this ensures the existence of an equivalent measure $\mathbb Q$ under which the model $f^\sigma = g(\cdot,Y^\sigma)$ is RN. Girsanov's theorem yields that 
	$$ dY^\sigma_t = \beta^{\sigma,\mathbb Q}_tdt + \sigma dW_t^{\mathbb Q},$$
	for some $\mathbb Q$-Brownian motion. In view of Proposition \ref{NA} we see that
	$$ d\bar{Y}_t = \beta^{\sigma}_tdt + \sigma dW_t,$$
	is then risk neutral under $\mathbb P$. This reveals that the existence of a $\mathbb P$-NA model represented with $g$ yields the existence of a $\mathbb P$-RN model represented with $g$ having the same constant diffusion coefficient.
	
	While the distinction between NA and RN does not play any role for (1) and (2) in the preceding Proposition \ref{NA} it has implications for (3) as these technical conditions would be required under the corresponding equivalent measure and the corresponding drift coefficient (for each possible choice of $\sigma$).
\end{rem}

In the SCC, we assume that for any given $\sigma$, the drift is simply a stochastic process. We would like to read Equation \eqref{eq:general} state by state. However, the given drift $\beta^\sigma$  is a stochastic process and does not need to be a function of the state. In the next two lemmas we show that there is a function $b$ such that Equation \eqref{eq:general} can be read state-by-state.

\begin{lem}\label{lem:drift}
	Let $f$ be an HJMM-model such that RN and SCC hold. For any $\sigma \in \mathbb{R}^{d\times d}$, let $\beta^\sigma$ be a drift corresponding to an RN and FDR model $f_t^\sigma(x)=g(x,Y_t^\sigma)$ with $dY_t^\sigma=\beta_t^\sigma dt+\sigma dW_t$. Let $t\geq 0$, $\sigma \in \mathbb{R}^{d\times d}$ and $b_t^\sigma:\mathbb{R}^d\rightarrow\mathbb{R}^d$ be measurable such that 
	$$b_t^\sigma(Y_t^\sigma):=\E [\beta_t^\sigma|Y_t^\sigma].$$

	Then we have 
	$$\partial_x g(x,y)=\nabla_y g(x,y)b_t^\sigma (y)+\frac12\sum_{i,j=1}^d \sigma_{i,j}\sigma_{j,i} \partial_{y_i}\partial_{y_j}g(x,y),$$
	for any $x\geq 0$ and any $(y,t)\in\mathbb R^d\times [0,\infty)$, $\mathbb P^{Y_t}\otimes \lambda$-a.s. 
\end{lem}
	\begin{proof}
		
		For $t,x\geq 0$ by conditioning on $Y_t^\sigma$ and using the formula from Proposition \ref{NA} we get
		\begin{align*}
		\partial_x g(x,Y_t^\sigma) dt=&\E\left[\mathrm{drift}(f(x))|Y_t^\sigma\right]\\
		=&\left\{\nabla_y g(x,Y_t^\sigma)b_t^\sigma(Y_t^\sigma)+\frac12\sum_{i,j=1}^d\sigma_{i,j}\sigma_{j,i} \partial_{y_i}\partial_{y_j}g(x,Y_t^\sigma)\right \}dt, \quad \mathbb P\text{-a.s}.
		\end{align*}		
		
	\end{proof}

\begin{lem}\label{lem:b2} 
	Let $f$ be an HJMM-model such that RN and SCC hold. Let $\sigma\in\mathbb R^{d\times d}$ such that the quadratic form $\sigma\sigma^\top$ is positive definite. Then there is a measurable function
	$$ b^\sigma:\mathbb R^d\rightarrow\mathbb R^d,$$ and a Lebesgue-null set $N\subseteq\mathbb R^d$ such that
	\begin{align}\label{sigmaid}
	\partial_x g(x,y) = \nabla_yg(x,y)b^\sigma(y) + \frac12\sum_{i,j=1}^d\sigma_{i,j}\sigma_{j,i} \partial_{y_i}\partial_{y_j}g(x,y),
	\end{align}
	for any $x\geq 0$ and any $y\in\mathbb R^d\setminus N$. 
	
	Moreover, $b^\sigma$ is locally bounded no matter the choice of $\sigma$.
\end{lem}	
	\begin{proof}
		Let $x\geq 0$, Lemma \ref{lem:drift} yields that there is a time $t>0$ such that
		$$\partial_x g(x,y)=\nabla_y g(x,y)b^{\sigma}(y)+\frac12\sum_{i,j=1}^d \sigma_{i,j}\sigma_{j,i} \partial_{y_i}\partial_{y_j}g(x,y), \quad \mathbb P^{Y_t}\text{-a.s.},$$
		where $b^\sigma :=b_t^\sigma$ and the latter is given in Lemma \ref{lem:drift}. 
		
		Lemma \ref{l:locb-ac} yields that the set 
		$$ N_x := \left\{ y\in\mathbb R^d: \partial_x g(x,y) \neq \nabla_yg(x,y)b^\sigma(y) + \frac12\sum_{i,j=1}^d\sigma_{i,j}\sigma_{j,i} \partial_{y_i}\partial_{y_j}g(x,y) \right\}, $$
		has zero Lebesgue measure for any $x\geq 0$. Define $N:= \bigcup_{x\geq 0,x\in\mathbb Q} N_x$ which satisfies $\lambda_d(N)=0$. The claim follows due to continuity in $x$.
		
	    By \eqref{locB} of Definition \eqref{SCC}, for any bounded set $B\subseteq\mathbb R^d$ there is $C>0$ such that $|\beta^{\sigma}_t|1_{\{Y^{\sigma}\in B\}} \leq C$, whence $b^\sigma \leq C\E[1_{Y^\sigma \in B}|Y_t^\sigma] \leq C $ which is locally bounded no matter the choice of $\sigma$.
	\end{proof}

Before presenting our main result, we recall quasi-exponential functions and some of their properties.

\begin{defn}[Quasi-exponential function]\label{df-QuasiExponential}
	A quasi-exponential (or {QE}) function is by definition of the form 
	\begin{align*}
	f(x)=\sum_j e^{\alpha x}\left[p_j(x)\cos(\omega_j x)+q_j(x)\sin(\omega_j x)\right],
	\end{align*}
	where $\alpha$, $\omega_j$ are real numbers, whereas $p_j,q_j$ are real polynomials.
\end{defn}

\begin{lem}\label{quasi-exponential}
	The following properties hold for the quasi-exponential functions:
	\begin{itemize}
		\item A function is QE if and only if it is a component of the solution of a vector-valued linear ODE with constant coefficients.
		\item A function is QE if and only if it can be written as $f(x) = c(e^{Ax}b)$ where $b,c\in\mathbb R^n$ and $A\in\mathbb R^{n\times n}$.
	\end{itemize}
	\begin{proof}
		See \cite[Lemma 2.1]{bjork}.
	\end{proof}
\end{lem}

Next we present our key result which will be used for the proof of our main result Theorem \ref{Main} below. It shows that if {SCC} holds, then $g$ has affine form. From a practical perspective, this means that the only “good” choices for $g$ are affine.

\begin{prop}\label{p:Main}
	Assume that RN holds (in the sense of Definition \ref{t:NA}) and that SCC holds (in the sense of Definition \ref{SCC}).
	
	Then the model has affine geometry in the sense of Definition \ref{d:flat}.
		
	Moreover, the corresponding functions $c:\mathbb R_+\rightarrow \mathbb R$, $u:\mathbb R_+\rightarrow \mathbb R^d$ are such that $c, u_1,\dots,u_d$ are QE.
\end{prop}
\begin{proof}
	
	For $i=1,\dots, d$ and denote by $e_{i,i}$ the $d\times d$-matrix with zeros everywhere except at position $(i,i)$ where it is $1$. By $e_{i,j}$ with $i,j=1,\dots,d$, we denote the $d\times d$-matrix with only zeroes except at position $(i,j)$ as well as at position $(j,i)$ where the entry is $1$. $I$ denotes the identity matrix.
	
	By Lemma \ref{lem:b2}, for any $\sigma\in\mathbb R^{d\times d}$ positive definite, there is a measurable function $b^\sigma:\mathbb R^d\rightarrow\mathbb R^d$ and a Lebesgue null set $N_\sigma\subseteq\mathbb R^d$ such that Equation \eqref{sigmaid} holds. In particular, we find via comparing with $\sigma=I$ and $\sigma= I+e_{i,i}$ that 
	$$ \partial_x g(x,y) = \nabla_y g(x,y)b^I(y)+\frac{1}2 \sum_{j=1}^d\partial_{y_j}^2g(x,y), $$
	$$ \partial_x g(x,y) = \nabla_yg(x,y)b^{I+e_{i,i}} + \frac32 \partial_{y_i}^2g(x,y) + \frac{1}2 \sum_{j=1}^d\partial_{y_j}^2g(x,y), $$
	for any $x\geq 0$ and any $y\in\mathbb R^d\setminus(N_I\cup N_{I+e_{i,i}})$. Then we have
	$$  \nabla_yg(x,y)b^{I} + \frac{1}2 \sum_{j=1}^d\partial_{y_j}^2g(x,y) = \nabla_yg(x,y)b^{I+e_{i,i}} + \frac32 \partial_{y_i}^2g(x,y) + \frac{1}2 \sum_{j=1}^d\partial_{y_j}^2g(x,y).$$
	
	Hence,
	$$ \partial_{y_i}^2g(x,y) = \nabla_yg(x,y)\left(\frac23(b^{I}-b^{I+e_{i,i}})\right) = \nabla_yg(x,y)\eta_{i,i}(y),$$
	where $\eta_{i,i}:\mathbb R^d\rightarrow \mathbb R^d$, $\eta_{i,i}(y) := \frac23(b^{I}-b^{I+e_{i,i}})$. $\eta_{i,i}$ is measurable in $y$ and $\eta_{i,i}$ is locally bounded by its construction as $b^I$ and $b^{I+ e_{i,i}}$ are locally bounded.
	
	Repeating for $\sigma = I$ and $\sigma = I + e_{i,j}$ where $i,j=1,\dots,d$ and $i\neq j$, we get the following result as before
	$$ \nabla_yg(x,y)b^{I} + \frac{1}2 \sum_{k=1}^d\partial_{y_k}^2g(x,y) = \nabla_yg(x,y)b^{I+e_{i,j}} + \partial_{y_i}\partial_{y_j}g(x,y) + \frac{1}2 \sum_{k=1}^d\partial_{y_k}^2g(x,y),$$
	
	Hence,
	$$ \partial_{y_i}\partial_{y_j}g(x,y) = \nabla_yg(x,y)\left(b^I-b^{I+e_{i,j}}\right) = \nabla_yg(x,y)\eta_{i,j}(y),$$
	where $\eta_{i,j}:= b^I-b^{I+e_{i,j}}$, $\eta_{i,j}$ is measurable in $y$ and $\eta_{i,j}$ is locally bounded by its construction as $b^I$ and $b^{I+e_{i,j}}$ are locally bounded.
	
	Using Equation \eqref{sigmaid} with $\sigma=I$ and multiplying it by $4$ and subtracting Equation \eqref{sigmaid} with $\sigma=2 I$ yields
	$$ 3\partial_xg(x,y) = 4\nabla_yg(x,y)b^{I} - \nabla_yg(x,y)b^{2I}  = 3\nabla_yg(x,y)\gamma(y),$$
	where $\gamma:\mathbb R^d\rightarrow \mathbb R^d, \gamma(y) := (4b^{I}-b^{2I})/3$.
	
	Hence we get for any $i,j=1,\dots,d$, $x\geq 0$, $y\in\mathbb R^d\setminus N$ where $N := N_I \cup N_{2I}  \cup \bigcup_{i,j=1}^d  N_{I+e_{i,j}}$
	\begin{align} 
	\partial_{y_i}\partial_{y_j}g(x,y) &= \nabla_yg(x,y)\eta_{i,j}(y), \label{eq: partial y}\\
	\partial_xg(x,y) &= \nabla_yg(x,y)\gamma(y). \label{eq:partial x}
	\end{align}
	
	Due to Equation \eqref{eq: partial y} and Proposition \ref{tcdf}, there exists a twice continuously differentiable function $A:\mathbb R^d\rightarrow \mathbb R^d$, where $A(0)=0$ such that
	$$ g(x,y) = g(x,0) + \nabla_y g(x,0) A(y),\quad  x \geq 0,\ y\in\mathbb R^d\setminus N . $$	
	
	We define $c:\mathbb R\rightarrow \mathbb R, c(x) :=  g(x,0)$, $u:\mathbb R\rightarrow \mathbb R^d, u(x) := \nabla_yg(x,0)$, $a:\mathbb R^d \rightarrow \mathbb R^{d \times d} $,  $a(y) :=D A(y)$, then
	\begin{align*}
	g(x,y) = c(x) + u(x)A(y),\quad  x \geq 0,\ y\in\mathbb R^d \setminus N. \end{align*}
    Note that $c,u$ are $C^1$-functions by definition.

	Let $V$ be the span of $\{ A(y) : y\in\mathbb R^d\setminus N\} $ and $k:= \dim(V)$. Let $\Pi: \mathbb R^d\rightarrow V$ be the orthogonal projection with respect to the standard scalar product. 
	Define $v(x):= \Pi\left(u(x)\right)$. Then
	\begin{align}\label{eq:v}
	g(x,y) = c(x) + v(x)A(y), \quad  x\geq 0,\ y\in\mathbb R^d \setminus N.
	\end{align}
	
	Consequently, $v$ is continuous. Equation \eqref{eq:v} yields that $v$ is continuously differentiable because $g$ and $c$ are continuously differentiable in $x$.
	
    We take the derivative of both sides of Equation \eqref {eq:v} with respect to $x$, and find by Equation \eqref{eq:partial x} that
	\begin{align}\label{eq:va}
	c'(x)+v'(x)A(y)=v(x)a(y)\gamma(y),\quad  x\geq 0,\ y\in\mathbb R^d\setminus N.
	\end{align}
	
	Let $U$ be the span of $\{ v(x) : x\geq 0 \}$ and $\Gamma:\mathbb R^d\rightarrow U$ be the orthogonal projection to $U$.
	
	Continuity of the left-hand side of Equation \eqref{eq:va} yields that
	$$w := \lim_{ y\rightarrow 0, y \in \mathbb R^d\setminus N} \Gamma\left(a(y)\gamma(y)\right),$$
	is well defined. Consequently, we have
	$$c'(x) = \lim_{y\rightarrow 0} \left(c'(x)+v'(x)A(y)\right) = v(x)w,$$
	and 
	$$c(x) = c(0) + \int_0^x v(z) dz w.$$
	
	Taking it into Equation \eqref{eq:va} we find that
	$$  v(x)w + v'(x)A(y) = v(x)a(y)\gamma(y), \quad x\geq 0,\  y\in\mathbb R^d\setminus N.$$
	
	There exist $y_1,\dots,y_k\in\mathbb R^d\setminus N$ such that $A(y_1),\dots, A(y_k)$ are orthogonal in $V$ and, hence, $A(y_1),\dots, A(y_k)$ is a basis of $V$.
	
	We have
	$$ v'(x) A(y_j) = v(x)a(y_j)\gamma(y_j) - v(x)w, \quad x\geq 0, j=1,\dots,k.$$
	
	By linearity there is a linear map $\tilde B:V \rightarrow V$ such that
	$$ v'(x) = \tilde Bv(x), \quad x\geq 0. $$
	
	Let $B:\mathbb R^d\rightarrow \mathbb R^d, y \mapsto \tilde B\Pi y $ which is a linear extension of $\tilde B$ to $\mathbb R^d$. Hence, we have
	$$ v'(x) = Bv(x), \quad x \geq 0. $$
	
	Lemma \ref{quasi-exponential} yields that $v_1,\dots,v_d$ are quasi-exponential functions. Since $c(x) = c(0) + \int_0^x v(z) dz w$ we find that $c$ is a quasi-exponential function.
	
	Hence, we have
	$$ g(x,y) = c(x) + v(x)A(y),\quad  x \geq 0, y\in\mathbb R^d, $$
	as claimed.
	
	Note that $u(x) = \nabla_yg(x,0) = DA(0)v(x)$ and $v(x) = \Pi(u(x)) = \Pi\left(DA(0)v(x)\right) = DA(0)v(x) = u(x)$ for any $x\geq 0$, i.e.\ $u = v$.	
\end{proof}

Observe that Equation \eqref{eq:affine geometry} is affine in $A(y)$. Therefore, if we replace the process $y$ by $A(y)$, then we get precisely an affine relation of the factor model.

\begin{cor}
	Assume the requirements of Theorem \ref{Main} and we use its notations. Then
	$Z_t := A(Y_t)$   is an It\^o process and we have
	$$f_t(x)  = c(x) + u(x) Z_t.$$
	\begin{proof}
		This is immediate from Corollary \ref{p:Main} and It\^o's formula.
	\end{proof}
\end{cor}

As soon as a parametrisation $g$ for the possible futures curves is chosen, it is desirable to estimate the volatility $\sigma$ of the underlying process $Y$ from data with some method. In order for this to make sense in view of absence of arbitrage, one needs that there is a corresponding drift coefficient such that 
 $$ f_t(x) = g(x,Y_t),\quad dY_t = \beta_t dt + \sigma_t dW_t,$$
 defines an NA model. Or equivalently, one can ask for the existence of the corresponding risk-neutral drift coefficient. Since we usually don't know which diffusion coefficient comes out of the estimation prior to inspecting the data, it is desirable to provide a drift coefficient for any possible diffusion coefficient. We start from this standpoint in our main theorem and prove that the corresponding model must have affine geometry.
\begin{thm}\label{Main}
	Assume that FDR and RN hold and that for any bounded progressively measurable $\mathbb R^{d\times d}$-valued process $\sigma$ there is a  progressively measurable $\mathbb R^d$-valued process $\beta^\sigma$ with locally integrable paths such that
		$$ dY^\sigma_t := \beta^\sigma_t dt + \sigma_t dW_t, \quad  f^\sigma_t(x) := g(x,Y^\sigma_t),$$
		defines an other RN model. We assume that $\beta^\sigma$ satisfies \eqref{SCC:def:3} of Definition \ref{SCC} if $\sigma$ is constant and deterministic.
	
	Then the model has affine geometry in the sense of Definition \ref{d:flat}.
\end{thm}
\begin{proof}
	The condition here implies the requirements of Proposition \ref{p:Main} and, hence, its conclusion follows.
	
%
%
%
%
\end{proof}

\section{Example}\label{s:example}
We present an example of a $1$-dimensional FDR\&RN model where the manifold of functions
$$\mathcal M = \left\{ g(\cdot,y) : y\in\mathbb R \right\},$$
is not affine-like, i.e., not contained in a finite dimensional space of curves, while the process $Y$ is simply a standard Brownian motion. Hence, 
$$f_t := g(\cdot ,Y_t),$$ 
is at any positive time near any given position in $\mathcal M$ with positive probability. 

Moreover, we embed the manifold into a Hilbert space and we show that the constructed process is a solution to the {HJMM}-equation, see Peszat and Zabczyk \cite[Theorem D.\ 2]{peszat2007stochastic} and the Hilbert space used is the one constructed in \cite[\textsection 5]{filipovic2001consistency}.

\begin{ex}\label{example2}
	We define
	$$g(x,y) = \Phi\left( \frac{1-y}{\sqrt{1+x}} \right), \quad  x\geq 0, y \in\mathbb R, $$
	where $\Phi$ is the distribution function of the standard normal law. Let $Y_t = Y_0 + W_t$, where $W$ is a $1$-dimensional standard Brownian motion with $Y_0\in\mathbb R$. We will also make use of the Hilbert space $H$ consisting of absolutely continuous functions $h:[0,\infty)\rightarrow \mathbb R$ which satisfy $\int_0^\infty |h'(x)|^2 \sqrt{(1+x)^3} dx <\infty$ endowed with the scalar product $$\langle h_1,h_2\rangle := h_1(0)h_2(0) + \int_0^\infty h_1'(x)h_2'(x) {(1+x)}^\frac{3}{2}dx.$$
	
	This Hilbert space has been discussed in \cite[\textsection 5]{filipovic2001consistency} and Equation (5.2) therein is not required due to Remark 5.1.1. The Hilbert space equipped with the scalar product is well defined.
	
	Essentially, we show that $f_t(x) := g(x,Y_t), t,x\geq 0$ defines an FDR\&RN model and that $f_t$ is an $H$-valued semimartingale in the sense of \cite[Definition 3.29]{peszat2007stochastic} which is a strong solution to the HJMM-equation
	$$ df_t = \partial_x f_t dt + \Sigma(f_t) dW_t, $$
	for some diffusion functional $\Sigma:H\rightarrow H$. We will also see that for no time $t>0$ there is a finite dimensional subspace $U$ of $H$ with $\mathbb P(f_t\in U)=1$. After this is shown it is clear from Theorem \ref{Main} that SCC does not hold because $g$ does not have the corresponding structure.
	\begin{itemize}
		\item \emph{FDR and RN condition hold.}	\\
		FDR holds by construction. 
		Moreover, we have 
		\begin{align}\label{x=frac12y}
		\partial_x g(x,y)=-\frac{1-y}{2 \sqrt{(1+x)^3}} \phi\left(\frac{1-y}{\sqrt{1+x}}\right)  =\frac12 \partial^2_y g(x,y),	
		\end{align}
		where $\phi$ is the normal density function and hence Proposition \ref{NA} yields RN.\\
		
		%
		
		\item \emph{$G:\mathbb R\rightarrow H,y\mapsto g(\cdot,y)$ is $C^2$.} \\
		$g$ is infinitely differentiable and, hence, $G(y)$ is absolutely continuous for any $y\in\mathbb R$. Moreover, we have
		\begin{align*}
		 \partial_xG(y) = \partial_xg(x,y) &= \phi\left(\frac{1-y}{\sqrt{1+x}}\right) \frac{y-1}{\sqrt{(1+x)^3}}\\ &= \frac{1}{\sqrt{2\pi}} \exp\left( -\frac{(1-y)^2}{2(1+x)}\right)\frac{y-1}{\sqrt{(1+x)^3}}, \end{align*}
		where $\phi$ denotes the density of the standard normal law. We find that
		$$\int_0^\infty |\partial _x G(y)|^2{\sqrt{(1+x)^3}}dx \leq \int_0^\infty \frac{|y-1|^2}{{\sqrt{(1+x)^3}}}dx < \infty, $$ 
		and, hence, $G(y)\in H$ for any $y\in\mathbb R$. The argument is easily repeated for $G'$ and $G''$, thus $G\in C^2(\mathbb R,H)$.\\
		
		\item \emph{$(f_t)_{t\geq 0}$ is an $H$-valued semimartingale.} \\
        For any $T\in L(H,\mathbb R)$ we have that $T\circ G$ is $C^2$ and It\^o's formula yields that
        \begin{align*}
        d T(f_t) = d(T\circ G)(Y_t) &= T(\frac12 G''(Y_t))dt+ T(G'(Y_t))dW_t \\
        &= T(\partial_xf_t) dt + T(G'(Y_t))dW_t . 
        \end{align*} 
    
         Observe that $t\mapsto \int_{0}^{t} \partial_x f_s ds$ is an $H$-valued process of bounded variation and $t\mapsto \int _0^t G'(Y_s) dW_s$ is an $H$-valued martingale. Consequently, $S_t := f_0 + \int_{0}^{t} \partial_x f_s ds + \int _0^t G'(Y_s) dW_s$ is some $H$-valued semi-martingale in the sense of \cite[Definition 3.29]{peszat2007stochastic} which satisfies 
          $$ dT(S_t) = T(\partial_xf_t)dt + T(G'(Y_t)) dW_t = dT(f_t), $$
          and, hence, $S_t = f_t$ by the Hahn-Banach separation theorem.
		
		
		\item \emph{$G$ is injective and $f$ solves the SPDE.} \\
		Let $y_1,y_2\in\mathbb R$ with $G(y_1) = G(y_2)$. Then $\Phi(1-y_1) = g(0, y_1) = g(0, y_2) = \Phi(1-y_2)$ and, hence, $y_1=y_2$. We denote by $G^{-1}$ the inverse of $G$ as a function to its range $R:=G(\mathbb R)$. Since $f_t = G(Y_t)$ we can apply $G^{-1}$ to $f_t$ for any $t\geq 0$.
		
		Define
		$$ \Sigma:H\rightarrow H, h \mapsto \begin{cases} G'(G^{-1}(h)) & h\in R, \\ 0 & \text{otherwise.}\end{cases} $$
		Then we have
		$$ df_t = \partial_xf_t dt + \Sigma(f_t) dW_t. $$
		
		\item \emph{The model is not affine-like.}\\
		Let $t>0$ and $U$ be any finite dimensional subspace of $H$. Note that the span of $R$ is $H$. Thus, there is $h\in R\setminus U$. Define $c := \mathrm{dist}(h,U)>0$. Then for some $\bar c>0$ we have
		$$\mathbb P(|f_t-h|< c) =\mathbb P(|Y_t-G^{-1}(h)|< \bar c) > 0,$$ and, hence, $\mathbb P(f_t\in U) < 1$ for any $t>0$. 
	\end{itemize}
	
\end{ex}

\begin{rem}
	From Equation \eqref{x=frac12y}, we can see the example above only works with {RN} under the condition that the volatility square is $1$. Due to this, it is incompatible with {SCC}. Moreover, in Theorem \ref{Main} we have shown that if {SCC} would hold, then the function $g$ would be affine as in Equation \eqref{eq:affine geometry}. 
	
\end{rem}






\appendix
\section{}
In the next proposition we consider a second order linear differential equation. Basically, we show that there is a solution independent function $A$ which is $C^2$ such that any $C^2$-solution can be represented by the function $A$, its initial value and its initial gradient. Note, that this does not ensure the existence of $C^2$-solutions but merely ensures that $C^2$-solutions have a common representation.

\begin{prop}\label{tcdf}
	Let $N\subseteq\mathbb R^d$ with $\lambda_d(N)=0$. Let $\eta_{i,j}:\mathbb R^d\setminus N\rightarrow \mathbb R^d$ be measurable and locally bounded with $\eta_{i,j}=\eta_{j,i}$ for any $i,j=1,\dots,d$.
	
	Then there is a twice continuously differentiable function $A:\mathbb R^d\rightarrow \mathbb R^d$ such that $A(0)=0$ and for any twice continuously differentiable function $g:\mathbb R^d\rightarrow \mathbb R$ with
	$$ \partial_i\partial_jg(y) = \nabla g(y) \eta_{i,j}(y),\quad i,j=1,\dots,d, y\in\mathbb R^d\setminus N, $$
	one has 
	$$ g(y) = g(0) + \nabla g(0) A(y),\quad y\in\mathbb R^d. $$
\end{prop}
\begin{proof}
	Define 
	\begin{align*}
	\mathcal S &:= \{ f\in C^2(\mathbb R^d,\mathbb R): \partial_i\partial_j f(y) = \nabla f(y) \eta_{i,j}(y),\text{ for any } i,j=1,\dots,d, y\in\mathbb R^d\setminus N \}, \\
	\mathcal S_0&:= \{ f \in \mathcal S: f(0) = 0 \}.
	\end{align*}
	Note that $\mathcal S$, $\mathcal S_0$ are vector spaces and for $f\in \mathcal S$ one has $h:=f-f(0)\in\mathcal S_0$ and
	$$ f = f(0) + h, $$
	i.e.\ $\mathcal S$ is the direct sum $\mathcal F_c \oplus \mathcal S_0$ where $\mathcal F_c$ is the space of constant functions. 
	
	Let $h \in \mathcal S_0$ with $\nabla h(0) = 0$. We show that $h=0$. To this end let $R>0$. By assumption on $\eta$ there is a constant $C\geq 0$ such that $\eta_{i,j}$ is bounded on the ball with radius $R$ by $C$  for any $i,j=1,\dots,d$. We find for $x\in\mathbb R^d\setminus \{0\}$ with $|x|\leq R$ that
	$$ |\nabla h(x)| = \left|\int_0^1 D(\nabla h)(tx) x dt\right| \leq Cd^2 \int_0^1 \left|\nabla h(tx)\right||x| dt. $$
	
	Thus, Gr\"onwall's lemma yields that
	$$ |\nabla h(x)| \leq |\nabla h(0)| \exp( |x| Cd^2) = 0, $$
	for any $x\in\mathbb R^d$ with $|x|\leq R$. Consequently, $\nabla h = 0$ which yields that $h$ is constant. Since $h(0)=0$ we find that $h=0$.
	
	Define $\Theta:\mathcal S_0\rightarrow\mathbb R^d, f\mapsto \nabla f(0)$. By the above we have that $\Theta$ is an injective linear map. Consequently, $l:= \dim(\mathcal S_0)\leq d$. Let $f_1,\dots,f_l$ be maximal linear independent in $\mathcal S_0$ such that $b^j:= \Theta(f_j)$, $j=1,\dots,l$ is an orthonormal system with respect to the standard scalar product. There is an orthogonal transformation $T:\mathbb R^d\rightarrow \mathbb R^d$ with $Tb^j = e_j$ for any $j=1,\dots, l$. We define
	\begin{align*}
	f(y) :=& (f_1(y),\dots, f_l(y),0,\dots, 0), \\
	A(y) :=& T^\top (f(y)),
	\end{align*}
	for any $y\in\mathbb R^d$. Note that $A:\mathbb R^d\rightarrow\mathbb R^d$ is twice continuously differentiable and $A(0)=0$. 
	
	Moreover, we have for $j=1,\dots,l$ and $y\in\mathbb R^d$ that
	$$ f_j(0) + \nabla f_j(0) A(y) = b^j T^\top (f(y)) = T(b^j) f(y) = e_j f(y) = f_j(y). $$
	By linearity we find that 
	$$ h(0) + \nabla h(0)A(y) = h(y), $$
	for any $y\in\mathbb R^d$, $h\in\mathcal S_0$. The claim follows.
\end{proof}

Next, we show that the Lebesgue measure (on $\mathbb R^d$) is absolutely continuous with respect to $\mathbb P^{X_t}$ for any $t>0$ where $X$ is an $\mathbb R^d$-valued diffusion with constant diffusion coefficient and locally bounded drift.
\begin{lem}\label{l:locb-ac}
	Let $\beta$ be an $\mathbb R^d$-valued progressively measurable process with locally integrable paths, $\sigma\in\mathbb R^{d\times d}$ be positive definite and $W$ be a $d$-dimensional standard Brownian motion. Let 
	$$ dX_t = \beta_t dt + \sigma dW_t ,$$
	and assume that for any bounded open set $U\subseteq\mathbb R$, we have $\beta_{\cdot} 1_{\{|X_{\cdot}|\in U \}}$ is a bounded process.
	
	Then for any $t>0$ and for any Borel $\mathbb P^{X_t}$-null set $N$ we have $\lambda_d(N)=0$.
\end{lem}
\begin{proof}
	Since $\Omega = \bigcup_{n\in\mathbb N} \{ |X_0|<n\}$ we may assume that $X_0$ is bounded.
	
	Let $t>0$ and $N\subseteq\mathbb R^d$ be a Borel $\mathbb P^{X_t}$-null set. Since any measurable set is a countable union of bounded measurable sets we may assume without loss of generality that $N$ is bounded. Then there is $K>0$ such that $|x|< K/2$ for any $x \in N \cup \mathrm{ran}(X_0)$ where $\mathrm{ran}(X_0):=\{X_0(\omega)|\omega\in\Omega\}$.
	
	The process $s\mapsto \beta_s 1_{\{|X_s|\leq K, s\leq t\}}$ is bounded by assumption. Hence, Girsanov's theorem yields a measure $\mathbb Q\sim \mathbb P$ such that there is a $\mathbb Q$-Brownian motion $B$ with
	\begin{align*}
	dX_s &= \beta^\mathbb Q_s ds + \sigma dB_s, \\
	\beta^\mathbb Q_s &= \beta_s 1_{\{|X_s|> K\text{ or }s>t\}}. 
	\end{align*}     
	
	Since $\mathbb Q$ is equivalent to $\mathbb P$ we have $\mathbb Q(X_t\in N) = 0$. We define $S := \{ \sup_{s\in[0,t]} |X_s| \leq K \}$ and have $1 > \mathbb Q(S) > 0$ because X has zero drift until its absolute value reaches $K$ and hence $X$ is a Brownian motion until then. Hence, we have	
	$$ 0 =  \mathbb Q( X_t \in N) = \mathbb Q( X_t \in N : S) \mathbb Q(S) + \mathbb Q( X_t \in N : S^c) \mathbb Q(S^c),$$
	where $\mathbb Q(A:B):= \frac{\mathbb Q(A \cap B)}{\mathbb Q(B)}$ is the conditional    probability.
	
	By positivity of the measure we find that $\mathbb Q( X_t \in N : S) \mathbb Q(S) = 0$ and since $\mathbb Q(S)>0$ we have $\mathbb Q( X_t \in N : S) = 0$. Since $X$ is a $\mathbb Q$-Brownian motion until its growth larger than $K$ we find that
	$$ 0 = \mathbb Q( X_t \in N : S) = \mathbb Q( X_0+\sigma B_t \in N : \sup_{s\in[0,t]} |X_0+\sigma B_s| \leq K ). $$
	
	Since $\sigma B$ is a non-degenerated $\mathbb Q$-Brownian motion, we find that $\lambda_d(N) = 0$ as claimed.
\end{proof}

\bibliographystyle{alpha}
\bibliography{REF}

\end{document}